\documentclass[conference]{IEEEtran}
\IEEEoverridecommandlockouts
\usepackage{booktabs}
\usepackage{latexsym}
\usepackage{cite}
\usepackage{graphicx}
\usepackage{eso-pic}
\usepackage{stfloats}
\usepackage{url}
\usepackage[table]{xcolor}
\definecolor{LLightBlue}{rgb}{0.9,0.95,1}
\usepackage{multirow}
\usepackage{amsmath}
\usepackage{amssymb}
\usepackage{amsthm}
\usepackage[ruled,vlined]{algorithm2e}
\newcommand{\modelname}{\textsc{HyperDoc}}
\newtheorem{theorem}{Theorem}
\newtheorem{definition}{Definition}

\begin{document}
\AddToShipoutPictureFG*{%
  \AtPageUpperLeft{%
    \put(\LenToUnit{0.55in},-\LenToUnit{0.24in}){%
      \parbox[t]{7.4in}{\centering\fontsize{8.5}{9.5}\selectfont
        \copyright~2027 IEEE. Personal use of this material is permitted.
        Permission from IEEE must be obtained for all other uses, in any current or future media,
        including reprinting/republishing this material for advertising or promotional purposes,
        creating new collective works, for resale or redistribution to servers or lists,
        or reuse of any copyrighted component of this work in other works.}%
    }%
  }%
}
\title{Hypergraph-Enhanced Evidence Selection for Long-Context Multimodal Document QA}

\author{
    \IEEEauthorblockN{
        Wenxin Zhao\IEEEauthorrefmark{1}, 
        Jiachuan Wang\IEEEauthorrefmark{2}\thanks{\IEEEauthorrefmark{2}Jiachuan Wang is the corresponding author.},
        Xinyi Zhu\IEEEauthorrefmark{1},
        Jiake Ge\IEEEauthorrefmark{3},
        Haoyang Li\IEEEauthorrefmark{4},
        Yongqi Zhang\IEEEauthorrefmark{1},
        Shuangyin Li\IEEEauthorrefmark{5},
        Lei Chen\IEEEauthorrefmark{1}\IEEEauthorrefmark{6}
    }
        
    \IEEEauthorblockA{
        \fontsize{10}{10}\selectfont{\IEEEauthorrefmark{1}The Hong Kong University of Science and Technology (Guangzhou), China} \\ 	
        \fontsize{9}{9}\selectfont{\ttfamily\upshape wzhao815@connect.hkust-gz.edu.cn, xzhu683@connect.hkust-gz.edu.cn, yongqizhang@hkust-gz.edu.cn} \\
        \fontsize{10}{10}\selectfont{\IEEEauthorrefmark{2}University of Tsukuba, Japan} \quad
        \fontsize{9}{9}\selectfont{\ttfamily\upshape wangjc@slis.tsukuba.ac.jp} \\
        \fontsize{10}{10}\selectfont{\IEEEauthorrefmark{3}Tianjin University, China} \quad \fontsize{9}{9}\selectfont{\ttfamily\upshape gejiake@tju.edu.cn} \\
        \fontsize{10}{10}\selectfont{\IEEEauthorrefmark{4}The Hong Kong Polytechnic University, Hong Kong SAR, China} \quad \fontsize{9}{9}\selectfont{\ttfamily\upshape haoyang-comp.li@polyu.edu.hk} \\
        \fontsize{10}{10}\selectfont{\IEEEauthorrefmark{5}South China Normal University, China} \quad \fontsize{9}{9}\selectfont{\ttfamily\upshape shuangyinli@scnu.edu.cn} \\
        \fontsize{10}{10}\selectfont{\IEEEauthorrefmark{6}The Hong Kong University of Science and Technology, Hong Kong SAR, China} \quad
        \fontsize{9}{9}\selectfont{\ttfamily\upshape leichen@cse.ust.hk} \\
    }
}

\maketitle

\begin{abstract}
Long-context multimodal document question answering (QA) has become an important yet challenging task for modern vision-language models (VLMs), as it requires reasoning over evidence dispersed across pages and modalities within limited input budgets. Current pipelines adopt retrieval-augmented generation (RAG) to select multimodal evidence for downstream VLMs. However, existing systems mainly rely on static selection rules over fixed-granularity retrieval units such as pages or blocks, resulting in information loss during retrieval and suboptimal performance in downstream reasoning. 
To alleviate these issues, we formulate the core task as budget-aware evidence selection over structured document evidence and propose \modelname{}, a document RAG framework that constructs compact yet sufficient evidence while preserving document relations. Specifically, \modelname{} constructs a document hypergraph to explicitly capture structural, logical, and cross-modal relationships among multimodal document elements. Over this structured evidence space, \modelname{} performs adaptive cost-aware hyperedge selection based on query intent and input budget. Extensive experiments on three representative benchmarks demonstrate that \modelname{} outperforms various types of state-of-the-art baselines.
\end{abstract}

\begin{IEEEkeywords}
Multimodal and heterogeneous data, budget-aware evidence selection, hypergraph, vision-language models
\end{IEEEkeywords}

\section{Introduction}
\label{sec:intro}
\vspace{-0.8mm}
In today's information landscape, a substantial portion of real-world information is stored in unstructured data~\cite{forbes}, with PDF documents serving as a major medium for preserving and communicating knowledge in domains such as scientific research, financial reporting, and technical support. Recent advances in large language models (LLMs) have improved their ability to query, extract, and reason over document content, but most LLM-powered systems can only process text representations parsed from documents~\cite{lin2024accurateefficientdocumentanalytics,linQueryingTemplatizedDocument2025,shankarDocETLAgenticQuery2025,DBLP:journals/pvldb/XiaoHWB25,10.1145/3786702}. However, real-world documents also encode meaningful information in multimodal content such as images and layouts~\cite{bhattacharyyaInformationExtractionVisually2025c, 10.1145/3626246.3654732}. Fortunately, with the rise of vision-language models (VLMs), recent work has expanded document QA toward visually rich, multi-page settings, where textual, visual, and structural evidence must be jointly preserved and integrated.

However, as a budget-aware evidence construction problem, long-context multimodal document QA is still challenging in two aspects.
First, although recent VLMs have demonstrated strong capabilities in straightforward QA tasks for single-page or short-form documents~\cite{DBLP:journals/corr/abs-2505-09388}, scaling such an end-to-end approach to entire long documents is often impractical. The finite context windows of VLMs, together with the memory and computation costs introduced by long document pages, pose budget limitations for evidence selection.
Second, evidence pieces distributed across pages and modalities also complicate evidence construction, as they must be selected and combined to support an answer. As shown in Figure~\ref{fig:preliminary_analysis}, a significant portion of questions in MMLongBench-Doc~\cite{maMMLongBenchDocBenchmarkingLongcontext2024b} and LongDocURL~\cite{dengLongDocURLComprehensiveMultimodal2025} require evidence spanning multiple pages, and the evidence is often distributed across different modalities, including text, layout, charts, tables, and images.

\begin{figure}[t]
    \centering
    \includegraphics[width=0.95\linewidth]{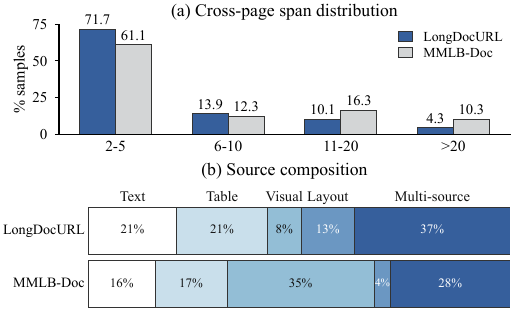}
    \vspace{-3mm}
    \caption{Evidence analysis on MMLongBench-Doc and LongDocURL.}
    \label{fig:preliminary_analysis}
    \vspace{-3mm}
\end{figure}

Existing methods mainly adopt multimodal Retrieval-Augmented Generation (RAG)~\cite{yu2025visragvisionbasedretrievalaugmentedgeneration, DBLP:journals/pvldb/WangWKGXC24,10.1145/3802111, ye-etal-2024-r2ag} to mitigate these difficulties. Instead of feeding the entire long document into a VLM, multimodal RAG first retrieves a query-relevant evidence subset, thereby reducing the pressure from limited context windows and inference budgets.
Specifically, to handle evidence distributed across pages and modalities, existing pipelines first convert each document into retrievable units for subsequent retrieval. They typically either use page images as natural retrieval units~\cite{suriVisDoMMultiDocumentQA2025b,hanMDocAgentMultiModalMultiAgent2025b,wuMoLoRAGBootstrappingDocument2025,choM3DocRAGMultimodalRetrieval2024b,DBLP:conf/iclr/ChenZZYDGRC025} or extract blocks as finer-grained chunks using Optical Character Recognition (OCR) techniques~\cite{gongMHierRAGMultiModalRAG2025,dongDocResearcherUnifiedSystem2025b}. However, this can break the hierarchical structure and logical relations within the document.
Then, during retrieval, most methods select the top-$K$ multimodal units using similarity scores produced by vision-based encoders such as ColPali~\cite{DBLP:journals/corr/abs-2407-01449} and CLIP~\cite{dblp:journals/corr/abs-2103-00020}. Even graph-based methods such as MoLoRAG~\cite{wuMoLoRAGBootstrappingDocument2025} ultimately select a fixed top-$K$ set of pages after traversal.
However, such a rigid cutoff is inherently suboptimal, as a large $K$ increases input cost and admits noisy evidence, while a small $K$ risks discarding evidence required for reasoning.
Thus, existing approaches still leave two issues insufficiently addressed for more adaptive and robust document evidence construction.

\begin{figure}[t]
    \centering
    \includegraphics[width=\linewidth]{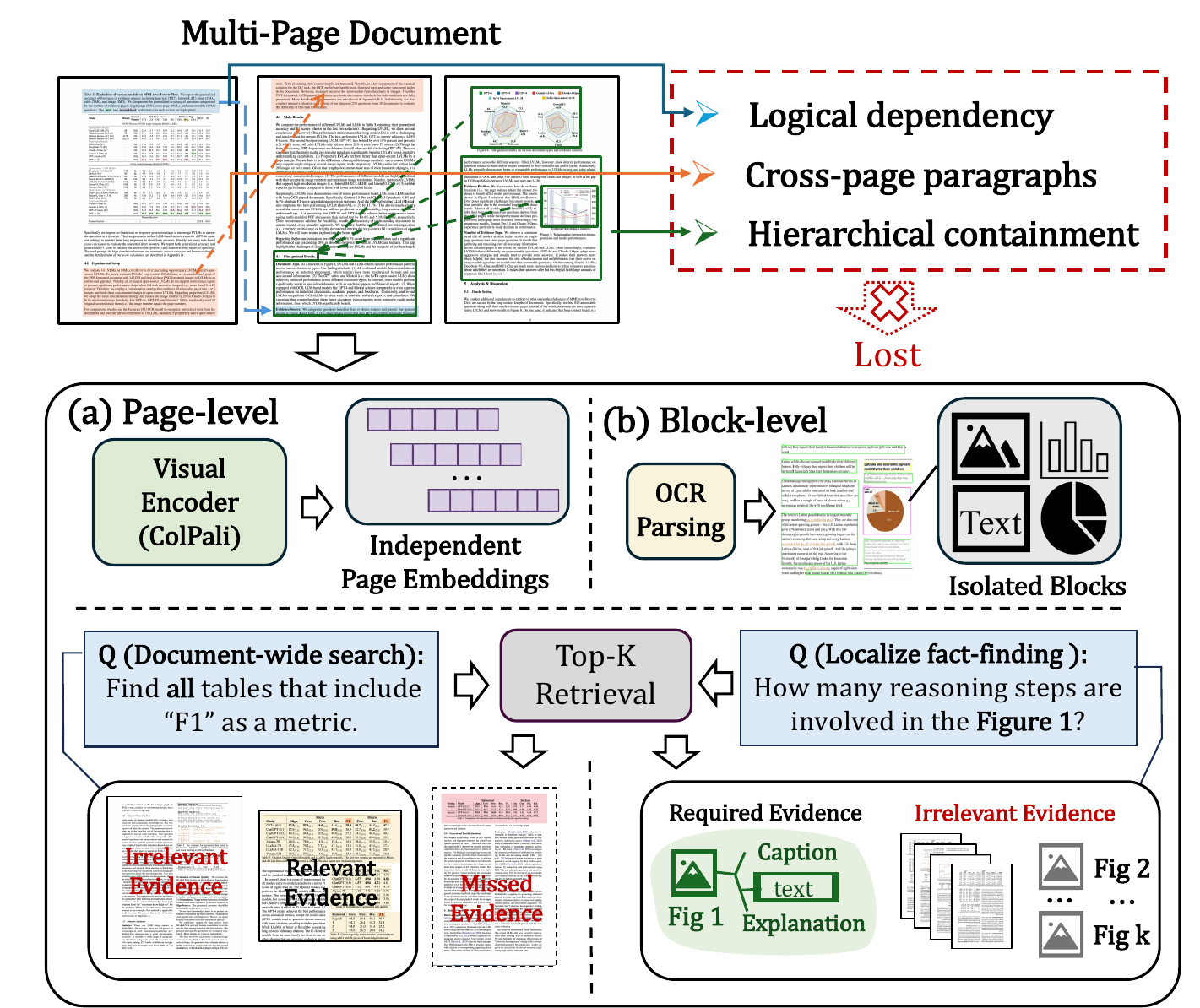}
    \vspace{-3mm}
    \caption{Limitations of existing document RAG pipelines.}
    \label{fig:motivation}
    \vspace{-3mm}
\end{figure}

The first challenge is the \textbf{\textit{relationally isolated evidence construction}} problem.
Multimodal document QA inherently requires reasoning over key information conveyed through structural and logical continuities, such as hierarchical containment and cross-page paragraphs, as shown in Figure~\ref{fig:motivation}.
For example, a query such as ``Which sections would you refer to for understanding the tasks involved in basic domestic cooking?’’ requires not only similarity-relevant content, but also document hierarchy information to identify the corresponding section titles.
Without explicitly capturing these relations from documents, existing methods often select isolated pages or blocks based on query similarity, which can lead to fragmented reasoning and information loss. Alternatively, they rely on LLM-generated summaries to recover missing connections, introducing additional computational cost~\cite{dongDocResearcherUnifiedSystem2025b}.
Such relations are difficult to extract because they often involve multiple heterogeneous elements jointly~\cite{Sarkhel2021ImprovingIE}, including section hierarchy and related multimodal context, while pairwise links represent such contexts only indirectly through multiple separate relations~\cite{yang2023hgmatchmatchbyhyperedgeapproachsubgraph}.
Thus, explicitly modeling document relations during evidence construction remains a key challenge for multi-page reasoning.

The second challenge is the \textbf{\textit{intent-agnostic evidence selection}} problem. Given a real-world document, the evidence required for different queries is highly heterogeneous, ranging from broad document-wide retrieval to localized fact retrieval as shown in Figure~\ref{fig:motivation}. However, existing approaches typically rely on a static, intent-agnostic evidence selection strategy such as top-$K$ similarity selection, making it difficult to adapt the evidence granularity and selection scope to the actual demands of different query intents.
Moreover, the retrieved evidence units are typically passed to the VLM without query-aligned organization, leaving the model itself to implicitly determine which evidence is most useful and how different pieces should be integrated.

In this paper, we address the aforementioned challenges by proposing a hypergraph-based document RAG framework \modelname{} designed to select evidence that is both query-relevant and relation-preserving under a limited budget.
First, we construct a document hypergraph offline to capture document structure and cross-modal contextual dependencies within each document.
Second, to mitigate the intent-agnostic evidence selection problem, we formulate a query-aware hyperedge selection problem at query time. The selected evidence is then pruned and linearized into a compact query-aligned context for downstream VLM reasoning.
The overall framework is shown in Figure~\ref{fig:framework}.

We summarize our main contributions as follows:
\begin{itemize}
\item We propose \modelname{}, a hypergraph-based document RAG framework that unifies structured evidence construction, intent-aware routing, and cost-aware selection.
\item To address relationally isolated evidence construction, we leverage a document hypergraph that explicitly models structural containment and cross-page contextual associations beyond fragmented page-level representations.
\item To address intent-agnostic evidence selection, we formulate query-aware hyperedge selection as a coverage maximization problem, prove its NP-hardness, and propose a cost-aware greedy algorithm that adapts selection to each query under limited input budgets.
\item Extensive experiments on representative benchmarks show that \modelname{} outperforms the state-of-the-art baselines across all backbone settings. Additional efficiency evaluations, ablation studies, and parameter sensitivity analyses further demonstrate its effectiveness, efficiency, and robustness.
\end{itemize}

\section{Preliminaries and Task Formulation}
\label{sec:problem}
\vspace{-0.8mm}
In this section, we introduce the basic notation used throughout the paper and define the evidence construction task in this work.

\begin{definition}[Hypergraph]
A hypergraph is denoted as $G=(V,\mathcal{E})$, where $V$ is a set of vertices and $\mathcal{E}$ is a set of hyperedges. Unlike a standard graph edge that connects two vertices, each hyperedge $e\in\mathcal{E}$ can connect an arbitrary subset of vertices, i.e., $e\subseteq V$.
\end{definition}

\begin{definition}[Document Hypergraph]
Given a long multimodal document $D$, we represent it as a document
hypergraph $\mathcal{G}_D=(V_D,\mathcal{E}_D)$. The node set $V_D$ contains
multi-granular document units, including page nodes and layout-aware
block nodes. Each hyperedge $e\in\mathcal{E}_D$ connects a set of units that
jointly form a document-native evidence relation, such as page-level containment,
section-level membership, or cross-modal context linking.
\end{definition}

\noindent\textbf{Task Formulation.}
Given a long multimodal document $D$, a user query $q$, and a VLM input budget $B_{\mathrm{vlm}}$, the task is to construct a compact VLM-ready multimodal input $E_q$ from $D$, consisting of serialized textual evidence and selected visual inputs, such that a VLM can answer $q$ using $E_q$ while satisfying the budget constraint $\mathrm{Cost}(E_q)\le B_{\mathrm{vlm}}$.

\section{Offline Document Hypergraph Construction}

\label{sec:construction}

\begin{figure*}[t]
    \centering
    \includegraphics[width=0.9\textwidth]{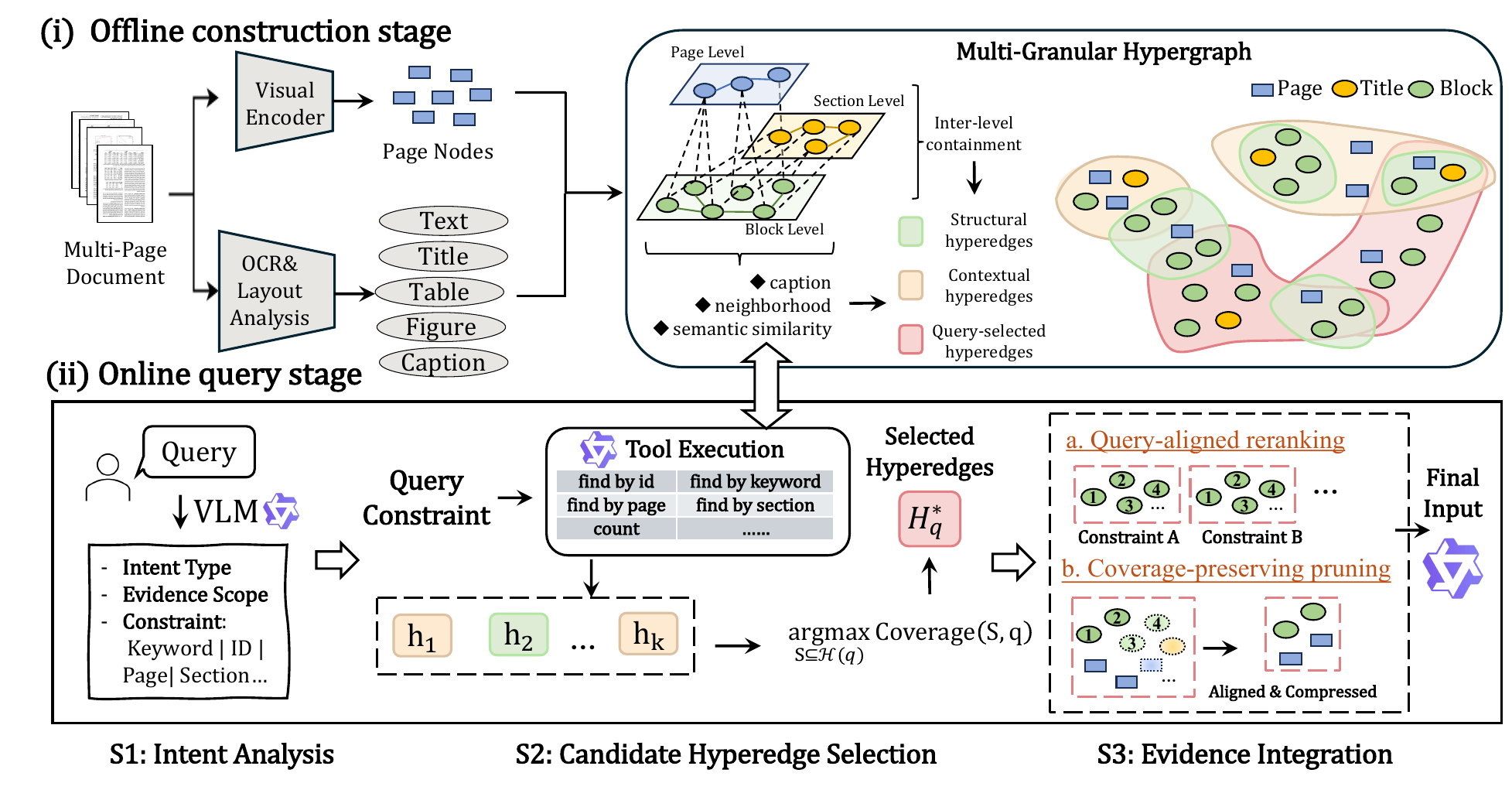}
    \vspace{-2mm}
    \caption{Overview of \modelname{}. The framework consists of two stages. In the offline stage, we construct a document hypergraph over multi-granular document elements using layout structure and semantic relations, so as to capture structural containment and cross-page contextual associations. In the online stage, we perform query-aware hyperedge selection to construct concise yet sufficient evidence for VLM reasoning.}
    \label{fig:framework}
    \vspace{-4mm}
\end{figure*}

In this section, we aim to address the \textbf{relationally isolated evidence
construction problem} by representing and preserving key relations within
a multimodal document as a document hypergraph.
Building such a document hypergraph requires resolving three representation challenges. First, multimodal documents contain heterogeneous evidence at different granularities. We parse these elements into unified hypergraph nodes while preserving both visual layout and textual content in Section~\ref{sec:document_eles}.
Second, constructing contextual relations requires identifying which heterogeneous elements jointly form a coherent evidence context. We address this by building contextual hyperedges in Section~\ref{sec:contextual_hyperedges}.
Third, document-level structural information, such as page containment and section membership, should be represented as selectable scopes that can guide query-time evidence selection. We therefore extract these relations in Section~\ref{sec:hyperedges}.

\subsection{Multimodal Document Parsing}
\label{sec:document_eles}
\vspace{-0.8mm}
Given a collection of multimodal documents $\mathcal{D}=\{D_1,\ldots,D_N\}$, we parse each document into pages and fine-grained blocks.
For a document $D=\{p_1,\ldots,p_n\}$, each page $p_j$ is rendered into an image $I_j$. The parsed page collection is represented as $P_D=\{(j,I_j)\}_{j=1}^{n}$, where $j$ is the page index. In implementation, page identifiers are normalized to zero-based indices. The page images retain visual layout, figures, tables, and other formatting cues, and later serve as coarse visual representations for candidate proposal and visual evidence rendering.

To capture finer-grained evidence, we apply a multimodal document parser, such as PaddleOCR~\cite{cui2025paddleocr30technicalreport}, to each page image and obtain a set of layout-aware blocks $B_D=\{b_1,\ldots,b_m\}$.
Each block is represented as
\begin{equation}
b_i=\langle \mathrm{id}_i,\tau_i,\mathrm{pg}_i,\mathrm{bbox}_i,\mathrm{text}_i\rangle,
\label{eq:block-representation}
\end{equation}
where $\mathrm{id}_i$ is the block identifier, $\tau_i$ denotes the element type, $\mathrm{pg}_i$ is the page index of the block, $\mathrm{bbox}_i$ specifies its page location, and $\mathrm{text}_i$ is the recognized textual content when available. The element types include text, title, list, figure, table, and caption blocks.
Among these parsed blocks, non-textual and layout-dependent elements are retained in their original visual form rather than being represented only by serialized text. Although elements such as tables and charts can often be converted into text, their spatial layout, graphical marks, and tabular structures may carry essential evidence that is difficult to preserve through serialization alone. We define the visual anchor set as
\begin{equation}
A_D=\{b_i\in B_D \mid \tau_i\in\mathcal{T}_{\mathrm{vis}}\},
\label{eq:visual-anchor-set}
\end{equation}
where $\mathcal{T}_{\mathrm{vis}}$ denotes visual or layout-dependent element types such as figure, table, chart, and diagram.
These parsed elements later serve as the vertices of
the document hypergraph and provide the basis for extracting contextual
and structural hyperedges.

\subsection{Visual-Centered Contextual Hyperedge Extraction}
\label{sec:contextual_hyperedges}
\vspace{-0.8mm}
In visually rich documents, visual elements such as figures and tables are often interpreted through multimodal context jointly formed by captions, titles, surrounding explanations, and referenced text.
To model these connections between multimodal evidence, we treat figures and tables as visual anchors used to gather their related textual context and extract the visual-centered n-ary relations as contextual hyperedges.

Given the visual anchor set $A_D$, we construct one contextual hyperedge
for each anchor $a_i\in A_D$. Specifically, we associate each visual
anchor with a caption set $\mathcal{C}_i^{cap}$, a section-title set $\mathcal{C}_i^{sec}$, and a set of related text blocks $R_i$.
Captions are selected from caption-type blocks using spatial proximity
and layout alignment with the anchor, while section titles are selected from preceding title blocks in reading order. Related text blocks are selected by
multimodal similarity, allowing us to recover explanatory content that may
not be adjacent to the visual element or explicitly marked as a caption.
Here, $\mathcal{C}_i^{cap}$ and $\mathcal{C}_i^{sec}$ contain the selected caption and section-title blocks, respectively, and are empty when no such block is available.

Let $\mathbf{v}_i$ be the visual embedding of anchor $a_i$, and let $\mathbf{t}_j$ be the textual embedding of a candidate text block $b_j$. We compute their relevance as
\[
\mathrm{sim}(a_i,b_j)=
\frac{\mathbf{v}_i^\top \mathbf{t}_j}
     {\|\mathbf{v}_i\|\,\|\mathbf{t}_j\|}.
\]
The top-ranked text blocks within a local page window are retained as $R_i$.

The visual-centered contextual hyperedge associated with anchor $a_i$ is defined as
\begin{equation}
  e_i^{ctx}
  =
  \{a_i\}
  \cup
  \mathcal{C}_i^{cap}
  \cup
  \mathcal{C}_i^{sec}
  \cup
  R_i ,
  \label{eq:contextual-hyperedge}
\end{equation}
where $e_i^{ctx}\in\mathcal{E}_D^{ctx}$. The set $\mathcal{E}_D^{ctx}$ contains all visual-centered contextual hyperedges extracted from document $D$.
Each contextual hyperedge binds a visual anchor with its related textual context, so that query-time selection can retrieve the visual anchor together with its supporting context as an indexable evidence group.

\subsection{Structural Hyperedge Extraction}
\label{sec:hyperedges}
\vspace{-0.8mm}

In this subsection, we extract structural hyperedges that expose hierarchical scopes over fine-grained evidence units.
Beyond multimodal connections among fine-grained blocks, documents naturally exhibit structural relations such as page containment, section hierarchy, and cross-page continuity.
At query time, a single evidence element may serve multiple structural roles simultaneously, acting as a page-level locator or a section-level contextual unit.
Since a single element can belong to multiple structural groups, pairwise graphs represent these groups indirectly through multiple edges; hyperedges provide a direct way to represent these group-level relations.

Specifically, we categorize and extract these structural hyperedges ($\mathcal{E}_D^{str}$) across two persistent structural levels:
\begin{equation}
    \mathcal{E}_D^{str} = \mathcal{E}_D^{page} \cup \mathcal{E}_D^{sec},
\label{eq:structural-edge-set}
\end{equation}
which represent structural relations at the page and section levels, respectively.

At the page level, we introduce page-containment hyperedges to preserve the physical boundaries of the source material. Let $v_k^p$ denote the page node for page $p_k$, and define $A_k^{page}=\{a_i\in A_D \mid \mathrm{pg}_i=k\}$ as the set of visual anchors located on page $p_k$, where $\mathrm{pg}_i$ is the page index inherited from the block representation in Eq.~\ref{eq:block-representation}:
\begin{equation}
    e_k^{page} = \{v_k^p\} \cup A_k^{page},
    \label{eq:page-edge}
\end{equation}
where $e_k^{page}\in\mathcal{E}_D^{page}$. This hyperedge exposes the page as a structural scope while preserving the fine-grained visual anchors within the same physical boundary.

Moving from physical layouts to semantic hierarchies, section-containment hyperedges model the logical organization of the document. Let $s_j$ denote a section title node, and let $A_j^{sec}$ denote the set of visual anchors falling within the textual scope governed by $s_j$. We define
\begin{equation}
    e_j^{sec}= \{s_j\} \cup A_j^{sec},
    \label{eq:section-edge}
\end{equation}
where $e_j^{sec}\in\mathcal{E}_D^{sec}$. This hyperedge binds the section heading with the multimodal evidence units governed by the same logical scope, allowing query-time selection to preserve section-level context.

Having extracted both contextual and structural relations, we integrate them to construct the unified document hypergraph:
\begin{equation}
    \mathcal{G}_D=(V_D,\mathcal{E}_D),
\label{eq:document-hypergraph}
\end{equation}
where the vertex set $V_D$ contains all parsed page and block elements, and the comprehensive hyperedge set is defined as the union of both relational modalities:
\begin{equation}
    \mathcal{E}_D = \mathcal{E}_D^{ctx} \cup \mathcal{E}_D^{str}.
\label{eq:unified-edge-set}
\end{equation}

Finally, this unified hypergraph paradigm organizes scattered document evidence into structurally interconnected representations rather than isolated pages or disjointed blocks. Consequently, query-time selection can efficiently route and retrieve evidence across varying scopes, from local multimodal contexts to global document structures.

\section{Online Query-Aware Hyperedge Selection}
\label{sec:selection}
\vspace{-0.8mm}
To resolve the \textbf{intent-agnostic evidence selection problem}, \modelname{} executes evidence selection as query processing over the pre-built, relation-aware document hypergraph $\mathcal{G}_D$.

\subsection{Overall Workflow}
\label{sec:selection_workflow}
\vspace{-0.8mm}
In this section, we describe how an input query is transformed into a query-specific evidence set for
downstream VLM answering. Formally, as illustrated in Figure~\ref{fig:framework}(ii), this process maps a query through
$q \rightarrow \pi(q) \rightarrow \mathcal{H}(q) \rightarrow H_q^* \rightarrow E_q,$
where $q$ is the user query, $\pi(q)$ is the parsed intent state, $\mathcal{H}(q)$ represents the candidate pool, $H_q^*$ denotes the selected evidence set, and $E_q$ signifies the final serialized VLM input. This workflow operates across three progressive stages:
\begin{table*}[!t]
\centering
\small
\setlength{\tabcolsep}{4pt}
\caption{Query intent types parsed in $\pi(q)$.}
\label{tab:intent-types}
\vspace{-1.5mm}
\begin{tabular}{@{}p{0.14\textwidth}p{0.34\textwidth}p{0.40\textwidth}@{}}
\toprule
\textbf{Intent} & \textbf{Description} & \textbf{Example} \\
\midrule
\textsc{locate} & Retrieve a localized evidence context anchored by an ID, page, section, or entity cue. & \emph{In the figure at the top of page 5, what is the final step?} \\
\textsc{count} & Count matching evidence instances across the document or a broad region. & \emph{How many tables include ``F1'' as a metric?} \\
\textsc{enumerate} & Collect and list matching items, attributes, or entities from the document. & \emph{List all titles of the charts in part 4.} \\
\textsc{compare} & Compare multiple evidence contexts to derive a difference, ranking, or relation. & \emph{Did more respondents say they are Product/Brand Focused or Customer Focused?} \\
\textsc{cross\_reasoning} & Compose information from multiple evidence contexts to infer an answer. & \emph{What is the sum of the 2007--2008 paid-search conversions and the number of green bars in the launch figure?} \\
\bottomrule
\end{tabular}
\vspace{-3mm}
\end{table*}
\paragraph{Step 1: Intent Analysis}
Given a user query $q$, \modelname{} first performs VLM-based intent analysis to predict both the intent type and the target evidence scope. This step aims to turn an open-form natural language question into a machine-actionable intent state $\pi(q)$ via the \textbf{Query Intent Analysis} procedure (Section~\ref{sec:intent_parsing}). The intent state summarizes the expected evidence scope and query cues, such as mentioned entities, page ranges, sections, or figure/table identifiers, which guide subsequent candidate proposal.

\paragraph{Step 2: Candidate Hyperedge Selection}

The parsed intent representation $\pi(q)$ is subsequently executed to activate \textbf{Intent-Guided Candidate Proposal} (Section~\ref{sec:candidate_proposal}) tools over the document hypergraph to construct an initial candidate pool $\mathcal{H}(q)$. The system then invokes \textbf{Cost-Aware Hyperedge Selection} to choose a compact subset of candidate hyperedges under the downstream input budget. The formal objective and greedy selection procedure are introduced in Section~\ref{sec:cost_aware_selection}.

\paragraph{Step 3: Evidence Integration}
Finally, the selected set $H_q^*$ is refined and serialized into the final VLM input $E_q$. This step applies query-aligned reranking and block-level pruning to produce a compact visual-textual context for downstream reasoning.

\subsection{Query Intent Analysis}
\label{sec:intent_parsing}
\vspace{-0.8mm}

To make heterogeneous query requirements executable during retrieval, we introduce query intent analysis. As summarized in Table~\ref{tab:intent-types}, real-world queries vary in both operation type and evidence scope, ranging from localized structural anchoring to document-wide aggregation and cross-page reasoning. Without explicitly parsing the query intent, evidence selection may either over-select irrelevant context for localized queries or under-select complementary evidence for broader queries.

Given the query $q$, we leverage a VLM to directly parse its requirements into three complementary operational signals. First, the intent type $t$ determines the high-level retrieval behavior, such as localization, counting, enumeration, comparison, or cross-evidence reasoning, as summarized in Table~\ref{tab:intent-types}. Second, the scope label $\ell\in\{\textsc{single},\textsc{multi},\textsc{document-wide}\}$ specifies the expected evidence range, indicating whether retrieval should focus on a single evidence group, compose multiple complementary groups, or aggregate evidence at the document level. Third, the query is decomposed into hard and soft constraints, which separate explicit structural requirements from semantic matching cues.

\begin{definition}[Query Constraints]
For a query $q$, let
\begin{equation}
C(q)=C_{\mathrm{hard}}(q)\cup C_{\mathrm{soft}}(q)
     =\{c_1,c_2,\ldots,c_m\},
\label{eq:query-constraints}
\end{equation}
denote its constraint set. Here, $C_{\mathrm{hard}}(q)$ contains explicit structural locators, such as page ranges, section mentions, and figure or table identifiers, while $C_{\mathrm{soft}}(q)$ contains semantic cues that guide keyword- or similarity-based matching.
\end{definition}

We then formalize these signals as the query-adaptive intent state
\begin{equation}
\pi(q)=\langle t,\ell,C_{\mathrm{hard}}(q),C_{\mathrm{soft}}(q)\rangle.
\label{eq:intent-state}
\end{equation}

The resulting $\pi(q)$ provides the routing specification for candidate proposal: it determines the activated proposal tools, while $C_{\mathrm{hard}}(q)$ restricts the admissible search space and $C_{\mathrm{soft}}(q)$ guides semantic recall expansion.

\subsection{Intent-Guided Candidate Proposal}
\label{sec:candidate_proposal}
\vspace{-0.8mm}

Given the parsed intent representation $\pi(q)$, this module activates a set of intent-guided proposal tools over the document hypergraph $\mathcal{G}_D$ and constructs a query-oriented candidate pool $\mathcal{H}(q)$ for later cost-aware selection.

Specifically, the router activates proposal tools according to the routing signals encoded in $\pi(q)$. Hard structural constraints, such as identifiers, page ranges, and section mentions, activate tools that retrieve structurally scoped evidence, while soft semantic cues, such as keywords, entities, or attributes, activate semantic matching tools to expand recall. In addition, \textsc{document-wide} queries with intents such as \textsc{count} and \textsc{enumerate} trigger document-scope tools that instantiate query-specific aggregation hyperedges over matched evidence units. Table~\ref{tab:tool_overview} summarizes these query-time proposal tools, their purposes, and the routing signals that activate them.

After routing, \modelname{} executes the activated tools over the document hypergraph $\mathcal{G}_D$ to produce the query-oriented candidate pool $\mathcal{H}(q)$.
Let $\mathcal{O}$ denote the full set of available proposal tools, and let $\Omega(\pi(q))\subseteq\mathcal{O}$ denote the subset activated by the parsed intent representation. The raw candidate set is first constructed as
\begin{equation}
\widetilde{\mathcal{H}}(q)
=
\bigcup_{\omega\in\Omega(\pi(q))}
\omega\bigl(
\mathcal{G}_D,
C_{\mathrm{hard}}(q),
C_{\mathrm{soft}}(q)
\bigr),
\label{eq:raw_candidate_pool}
\end{equation}
where $\omega$ denotes an atomic proposal tool. The constraint sets are not only used for routing, but also passed as execution inputs: $C_{\mathrm{hard}}(q)$ provides structural locators such as page ranges, section titles, or figure/table identifiers, while $C_{\mathrm{soft}}(q)$ provides semantic cues for keyword- or similarity-based matching.

\begin{table}[t]
\centering
\footnotesize
\setlength{\tabcolsep}{4pt}
\vspace{-1.5mm}
\caption{Query-time proposal tools and their routing signals.}\label{tab:tool_overview}
\vspace{-1.5mm}
\begin{tabular}{@{}p{0.27\columnwidth}p{0.29\columnwidth}p{0.34\columnwidth}@{}}
\toprule
\textbf{Tool} & \textbf{Purpose} & \textbf{Routing Signal} \\
\midrule
\texttt{find\_by\_id}
& Locate ID-anchored evidence.
& Explicit figure, table, or block ID in $C_{\mathrm{hard}}(q)$. \\
\texttt{find\_by\_page}
& Retrieve page-scoped evidence.
& Page index or resolved page range in $C_{\mathrm{hard}}(q)$. \\
\texttt{find\_by\_section}
& Retrieve section-scoped evidence.
& Section title or resolved section scope in $C_{\mathrm{hard}}(q)$. \\
\texttt{keyword\_search}
&  Construct keyword-matched candidate hyperedges.
& Keywords, entities, or attributes in $C_{\mathrm{soft}}(q)$. \\
\texttt{global\_count}
& Instantiate a query-specific counting hyperedge.
& Count intent $t$ with document-wide scope $\ell$. \\
\texttt{enumerate}
& Instantiate a query-specific enumeration hyperedge.
& Enumeration intent $t$ or multi-item scope $\ell$. \\
\texttt{visual\_search}
& Supply cross-modal page-level evidence candidates.
& Weak textual cues or visual-prior fallback. \\
\bottomrule
\end{tabular}
\vspace{-3mm}
\end{table}

The outputs of all selected tools are then deduplicated by hyperedge identity or page identity, yielding the candidate pool
\begin{equation}
\mathcal{H}(q)
=
\operatorname{Dedup}\bigl(\widetilde{\mathcal{H}}(q)\bigr).
\label{eq:candidate_pool}
\end{equation}

The resulting $\mathcal{H}(q)$ is therefore an over-complete candidate pool containing both pre-built hyperedges retrieved from the offline document hypergraph and query-induced hyperedges instantiated during online proposal, such as keyword-matched groups or aggregation hyperedges for counting and enumeration. This design allows \modelname{} to go beyond pre-built document relations and construct evidence groups specified by the query semantics.

\subsection{Cost-Aware Hyperedge Selection}
\label{sec:cost_aware_selection}
\vspace{-0.8mm}
Given the candidate pool $\mathcal{H}(q)$, the selector aims to retain a compact subset of hyperedges that sufficiently covers the query intent under the VLM input budget. To address this budget-aware evidence construction problem introduced in Section~\ref{sec:intro}, we formulate query-aware hyperedge selection as a coverage maximization problem under a VLM-budget-derived cardinality cap and show that it is NP-hard. Based on this formulation, we introduce a cost-aware greedy algorithm that ranks candidates by marginal coverage gain per unit input cost and stops adaptively when additional evidence is no longer cost-effective. We next define the coverage objective and input-cost signal before presenting the greedy selection algorithm.

\begin{definition}[Intent-Conditioned Coverage]
Let $\mathcal{H}(q)=\{h_1,\ldots,h_n\}$ be the candidate hyperedge pool for query $q$. For any selected subset
$S\subseteq\mathcal{H}(q)$ and query intent $t$, we define its coverage as
\begin{equation}
\mathrm{Coverage}_t(S,q)
=
\bigl(
\mathrm{Cov}_{\mathrm{hard}}(S,q),
\mathrm{Cov}^{t}_{\mathrm{soft}}(S,q)
\bigr),
\label{eq:intent-coverage}
\end{equation}
where the two components are compared lexicographically. The hard component evaluates whether the selected hyperedges satisfy the explicit structural constraints in $C_{\mathrm{hard}}(q)$, such as identifiers, pages, sections, or document-wide aggregation requirements. The intent-conditioned soft component evaluates how well $S$ supports the semantic requirements of $q$ under intent $t$.
\end{definition}

\begin{definition}[Input Cost]
For a selected hyperedge set $S$, we define $\mathrm{Cost}(S)$ as the estimated
input cost after rendering $S$ into the VLM input. Let $\mathcal{M}(S)$ denote the
deduplicated set of visual-textual evidence units included by $S$, including
rendered page regions, visual crops, textual blocks, and structural metadata.
Then,
\[
\mathrm{Cost}(S)
=
\sum_{u\in \mathcal{M}(S)} \kappa(u),
\]
where $\kappa(u)>0$ is the estimated input cost of evidence unit $u$, measured
by the textual tokens and visual tokens it contributes after preprocessing.
\end{definition}

For a candidate hyperedge $h$, its marginal input cost with respect to the current selected set $S$ is defined as
\begin{equation}
\Delta\mathrm{Cost}(h\mid S)
=
\mathrm{Cost}(S\cup\{h\})-\mathrm{Cost}(S).
\label{eq:marginal-cost}
\end{equation}
With the coverage objective defined, we formulate query-aware hyperedge selection as a lexicographic coverage maximization problem.

\noindent\textbf{Problem Definition.}
Given the document hypergraph $\mathcal{G}_D$, a query $q$, and a candidate pool $\mathcal{H}(q)$, the problem is
to select an evidence set $H_q^*\subseteq\mathcal{H}(q)$ that satisfies explicit query constraints and maximizes intent-conditioned semantic coverage. Formally, we define
\begin{equation}
\begin{aligned}
H_q^*
=
\arg\max_{\substack{
S\subseteq\mathcal{H}(q)\\
|S|\le K}}
\Bigl(
\mathrm{Cov}_{\mathrm{hard}}(S,q),
\mathrm{Cov}^{t}_{\mathrm{soft}}(S,q)
\Bigr),
\end{aligned}
\label{eq:selection-objective}
\end{equation}
where the two terms are compared lexicographically. The selection process first prioritizes explicit constraint satisfaction and then maximizes intent-conditioned semantic support. Here, $K$ is not a manually specified retrieval cutoff, but a budget-derived cardinality upper bound computed as $K=\lfloor B_{\mathrm{vlm}}/c_{\max}\rfloor$, where $B_{\mathrm{vlm}}$ is the model-specific input budget and $c_{\max}$ is an offline upper bound on the rendered cost of one hyperedge. This design keeps the formal search space compatible with the model budget, avoiding contexts that exceed the maximum context window or cause inference-time errors. Note that K only serves as a feasibility cap rather than a target selection size. The actual number of selected hyperedges is determined adaptively by the cost-aware stopping rule used in the greedy selection procedure.

\begin{theorem}
The query-aware hyperedge selection problem is NP-hard.
\end{theorem}

\begin{proof}
We prove NP-hardness by reduction from Maximum Coverage. Consider a
restricted instance with no hard constraints,
$C_{\mathrm{hard}}(q)=\emptyset$, and only soft constraints. Given a
Maximum Coverage instance with universe
$U=\{u_1,\ldots,u_m\}$, subset collection
$\mathcal{A}=\{A_1,\ldots,A_n\}$, and budget $K$, construct a query $q$
with
\[
C_{\mathrm{soft}}(q)=\{c_1,\ldots,c_m\},
\qquad
C_{\mathrm{hard}}(q)=\emptyset,
\]
where each $c_j$ corresponds to $u_j$.

For each subset $A_i\in\mathcal{A}$, create a candidate hyperedge
$h_i\in\mathcal{H}(q)$. Define a binary soft matching function
$f(c_j,h_i)\in\{0,1\}$ such that
\[
f(c_j,h_i)=1
\quad\text{if and only if}\quad
u_j\in A_i.
\]
In this restricted instance, $\mathrm{Cov}_{\mathrm{hard}}(S,q)$ is constant for all feasible $S$. The soft coverage objective then becomes
\[
\mathrm{Cov}^{t}_{\mathrm{soft}}(S,q)
=
\sum_{c_j\in C_{\mathrm{soft}}(q)}
\max_{h_i\in S} f(c_j,h_i),
\]
which is exactly the number of universe elements covered by the selected
subsets. Therefore, solving this restricted instance is equivalent to
selecting at most $K$ sets from $\mathcal{A}$ to maximize coverage over
$U$. Since Maximum Coverage is NP-hard, the general query-aware hyperedge
selection problem is also NP-hard.
\end{proof}

Since the query-aware selection problem is NP-hard, we use a greedy density-based selection strategy, as described in Algorithm~\ref{alg:greedy_hyperedge}.
The greedy strategy approximates the lexicographic structure of Eq.~\ref{eq:selection-objective}. It first prioritizes structurally admissible evidence and then performs cost-aware semantic augmentation.
First, the selector ensures hard coverage by retaining candidates that directly satisfy the explicit structural constraints in $C_{\mathrm{hard}}(q)$, such as identifier, page, or section constraints.
This anchors the selected evidence to the admissible structural scope specified by the query.
Then, the selector evaluates candidates retrieved by semantic cues in $C_{\mathrm{soft}}(q)$ to improve soft coverage under the input budget.
Let $c_i\in C_{\mathrm{soft}}(q)$ denote a soft constraint and let $f(c_i,h)\in[0,1]$ denote its matching score with candidate hyperedge $h$.
The marginal soft gain $\Delta_{\mathrm{soft}}(h\mid S,q)$ measures how much additional soft-constraint coverage $h$ contributes under the parsed evidence scope $\ell$.
Formally, for \textsc{single} queries, we use the direct soft coverage of $h$:
\begin{equation}
\Delta_{\mathrm{soft}}(h\mid S,q)
=
\frac{1}{|C_{\mathrm{soft}}(q)|}
\sum_{c_i\in C_{\mathrm{soft}}(q)} f(c_i,h).
\label{eq:single-soft-gain}
\end{equation}
For \textsc{multi} queries, we only count the improvement beyond the current selected set $S$:
\begin{equation}
\Delta_{\mathrm{soft}}(h\mid S,q)
=
\frac{1}{|C_{\mathrm{soft}}(q)|}
\sum_{c_i\in C_{\mathrm{soft}}(q)}
\max\bigl(0, f(c_i,h)-\beta_i(S)\bigr),
\label{eq:multi-soft-gain}
\end{equation}
where $\beta_i(S)=\max_{h'\in S}f(c_i,h')$ denotes the best current support for soft constraint $c_i$. Thus, \textsc{single} queries favor strong direct matches, while \textsc{multi} queries favor complementary candidates that add new or stronger support for the soft cues.

To account for input cost, candidates are ranked by their marginal utility density. Since every candidate hyperedge contributes at least one rendered evidence
unit with positive cost, $\Delta\mathrm{Cost}(h\mid S)>0$.
The density is defined as
\begin{equation}
\rho(h\mid S,q)
=
\frac{
\Delta_{\mathrm{soft}}(h\mid S,q)
}{
\Delta\mathrm{Cost}(h\mid S)
}.
\label{eq:greedy-density}
\end{equation}
This density favors cost-effective candidates and penalizes \textbf{noisy evidence} that increases the input cost without sufficient marginal query coverage.
The selector continues soft augmentation while
\begin{equation}
\rho(h^\ast\mid S,q)\ge \alpha_{\rho}\rho_1,
\label{eq:density-threshold}
\end{equation}
where $\rho_1$ is the density of the first selected soft hyperedge. When the best remaining candidate falls below this threshold, the selector stops because no remaining evidence is cost-effective enough relative to the strongest evidence for the query. Here, $K$ serves as a model-budget-derived feasibility cap rather than a fixed retrieval cutoff. In practice, the hard-coverage initialization is typically far smaller than $K$. The effective selection size is determined adaptively by the density-based stopping rule in Eq.~\ref{eq:density-threshold}, and the selector terminates before reaching $K$ in most cases.

\begin{algorithm}[!htbp]
\small
\caption{Cost-Aware Hyperedge Selection}
\label{alg:greedy_hyperedge}
\KwIn{Candidate pool $\mathcal{H}(q)$, query $q$, parsed intent $\pi(q)$, VLM-budget-derived cap $K$}
\KwOut{Selected evidence set $H_q^*$}

Extract evidence scope $\ell$, hard constraints $C_{\mathrm{hard}}(q)$, and soft constraints $C_{\mathrm{soft}}(q)$ from $\pi(q)$\;

\tcp{Hard-coverage initialization}
$S \leftarrow \{h\in\mathcal{H}(q)\mid \operatorname{SatisfyHard}(h,C_{\mathrm{hard}}(q))\}$\;

\tcp{Density-based soft augmentation}
$P \leftarrow \mathcal{H}(q)\setminus S$\;
$\rho_1 \leftarrow \bot$\;

\While{$P\neq\emptyset$ \textbf{and} $|S|<K$}{
    \ForEach{$h\in P$}{
        compute $\Delta_{\mathrm{soft}}(h\mid S,q)$ according to $\ell$\;
        compute $\rho(h\mid S,q)$ using Eq.~\ref{eq:greedy-density}\;
    }
    choose $h^*$ by the highest $\rho(h\mid S,q)$\;
    \If{$\rho_1\neq\bot$ \textbf{and} $\rho(h^\ast\mid S,q)<\alpha_{\rho}\rho_1$}{
        \textbf{break}\;
    }
    \If{$\rho_1=\bot$}{
        $\rho_1\leftarrow \rho(h^*\mid S,q)$\;
    }
    $S\leftarrow S\cup\{h^*\}$\;
    $P\leftarrow P\setminus\{h^*\}$\;
}

\Return{$S$}\;
\end{algorithm}

After hyperedge selection, the selected set $H_q^*$ is further refined by two lightweight evidence-integration modules. First, a \emph{constraint-aligned reranking
module} reorders the selected hyperedges according to their matched constraints
and evidence scope, so that evidence supporting the same constraint, anchor
entity, or sub-question is presented in a more coherent sequence. This is
particularly useful for \textsc{multi} queries, where multiple selected
hyperedges may correspond to different entities or constraint groups and should
be organized before being rendered to the VLM.

Second, a \emph{block-level pruning module} further compresses the selected evidence after hyperedge selection.
Inside this module, we perform a second-stage greedy refinement at the block level. Given a selected soft hyperedge, we treat its anchor and caption blocks as fixed structural seeds, and greedily retain only auxiliary related blocks whose marginal coverage gain justifies their additional context cost. In this way, pruning compresses noisy local context inside each hyperedge without breaking the structural faithfulness of the selected evidence.

The final evidence context rendered to the VLM is then obtained from the union
of the refined selected hyperedges:
\begin{equation}
E_q = \bigcup_{h \in \widetilde{H}_q^*} h,
\label{eq:final-evidence}
\end{equation}
where $\widetilde{H}_q^*$ denotes the selected set after reranking and pruning.
This query-dependent evidence context can retain broader page- or section-level
evidence for document-wide queries, while collapsing to a small number of
visual-centered contextual hyperedges for localized queries. As a result, the
VLM receives evidence that is more coherent than disconnected OCR blocks and
more compact than full-page retrieval.

\section{Experiments}
\label{sec:experiments}
\vspace{-0.8mm}
In this section, we evaluate the overall performance of \modelname{} on long-context multimodal document QA tasks compared to the state-of-the-art baselines. We also conduct an efficiency analysis, an ablation study, and a parameter sensitivity analysis, and present a case study.
\subsection{Experiment Setup}
\vspace{-0.8mm}
\subsubsection{Datasets}
\label{sec:experimental_datasets}
To assess the performance of \modelname{} on long-form document QA, we evaluate it on three representative benchmarks: MMLongBench-Doc (MMLB-Doc)~\cite{maMMLongBenchDocBenchmarkingLongcontext2024b}, LongDocURL (LDU)~\cite{dengLongDocURLComprehensiveMultimodal2025}, and DocBench~\cite{zouDOCBENCHBenchmarkEvaluatingb}. They cover diverse domains, including administrative files, tutorials, and research reports, and feature multimodal content such as charts, tables, and text. Dataset statistics are summarized in Table~\ref{tab:longdoc_benchmarks_transposed}. All benchmark documents are digital PDFs, while scanned PDFs can be processed by the same rendering, parsing, and hypergraph construction pipeline without modification. With average lengths of 47.5--85.6 pages, these benchmarks require evidence construction beyond direct VLM inference. Datasets such as DocVXQA~\cite{souibguiDocVXQAContextAwareVisualb} and DocVQA~\cite{mathew2021docvqadatasetvqadocument} are therefore excluded from the main comparison, as they primarily focus on short documents or single-page QA and are less aligned with the long-context setting targeted in this work.
\begin{table}[h]
\centering
\small
\setlength{\tabcolsep}{4pt}
\vspace{-2mm}
\caption{Dataset statistics of the evaluated document QA benchmarks.}
\label{tab:longdoc_benchmarks_transposed}
\vspace{-1.5mm}
\begin{tabular}{lccc}
\toprule
Metric & MMLongBench-Doc & LongDocURL & DocBench \\
\midrule
\# Questions & 1091 & 2325 & 1102 \\
\# Documents & 135 & 396 & 229 \\
Avg. Pages & 47.5 & 85.6 & 66.0 \\
Avg. Tokens & 21,214.1 & 56,715.1 & 46,377.0 \\
\bottomrule
\end{tabular}
\vspace{-2mm}
\end{table}
\subsubsection{Metrics}
\label{sec:experimental_metrics}
We follow the official evaluation metric of each benchmark~\cite{maMMLongBenchDocBenchmarkingLongcontext2024b,dengLongDocURLComprehensiveMultimodal2025,zouDOCBENCHBenchmarkEvaluatingb}. Specifically, we report generalized accuracy (ACC) and generalized F1 for MMLongBench-Doc and LongDocURL, and report DocBench accuracy under its LLM-based semantic judging protocol using GPT-4o~\cite{openai2024gpt4ocard}.

\subsubsection{Baselines}
\label{sec:experimental_baselines}
To comprehensively evaluate our framework, we compare it against representative baselines that construct evidence for long-document QA through textual or multimodal retrieval.
For \emph{(1) textual input methods}, we include methods that rely on OCR-parsed or structurally extracted text as the primary input without directly using rendered visual evidence, represented by DocETL~\cite{shankarDocETLAgenticQuery2025}.
For \emph{(2) multimodal RAG methods}, we include methods that retrieve visual or multimodal evidence from documents, including the page-retrieval baseline M3DocRAG~\cite{choM3DocRAGMultimodalRetrieval2024b} with different retrieved-page budgets, the multi-agent baseline MDocAgent~\cite{hanMDocAgentMultiModalMultiAgent2025b}, the graph-retrieval baseline MoLoRAG~\cite{wuMoLoRAGBootstrappingDocument2025}, and MHier-RAG~\cite{gongMHierRAGMultiModalRAG2025}.

\subsubsection{Implementation Details}
\label{sec:implementation_details}

To ensure a fair and rigorous evaluation, we implement \modelname{} and all comparative baselines within a standardized and reproducible pipeline. For VLM inference, we run Qwen-family VLMs~\cite{DBLP:journals/corr/abs-2505-09388,DBLP:journals/corr/abs-2502-13923} on NVIDIA H100 GPUs. For data preprocessing, we use PaddleOCR~\cite{cui2025paddleocr30technicalreport} to extract layout-aware document blocks and ColPali v1.3~\cite{DBLP:journals/corr/abs-2407-01449} for page-level retrieval.
Unless specified otherwise, the density stopping coefficient $\alpha_{\rho}$ is set to 1.0 (see Section~\ref{sec:parameter_sensitivity} for a detailed analysis). The VLM input budget $B_{\mathrm{vlm}}$ is determined by the maximum context window encoded in the model configuration, e.g., 128K tokens for Qwen2.5-VL-7B and 256K tokens for Qwen3-VL-8B/32B, and is shared by textual and visual tokens after preprocessing. The complete source code, prompts, and configurations are available at \url{https://github.com/WenxinAZhao/HyperDoc}.

\subsection{Overall Performance}
\label{sec:overall_performance}
\vspace{-0.8mm}
In this subsection, we present a comprehensive evaluation of \modelname{}. We evaluate its overall question-answering accuracy, retrieval effectiveness, and inference efficiency against state-of-the-art baselines, as well as its intent-type classification accuracy on manually labeled queries.
\begin{table*}[!t]
\centering
\setlength{\tabcolsep}{16pt}
\caption{Overall comparison on different backbone models. Bold and underlined values denote the best and second-best results within each backbone, respectively; blue cells highlight \modelname{}.}
\label{tab:overall_results}
\vspace{-2mm}
\small
\resizebox{\textwidth}{!}{\renewcommand{\arraystretch}{0.88}%
\begin{tabular}{l l | cc | cc | c}
\toprule
\multirow{2}{*}{Backbone Model} &
\multirow{2}{*}{Method} &
\multicolumn{2}{c|}{MMLongBench-Doc} &
\multicolumn{2}{c|}{LongDocURL} &
\multicolumn{1}{c}{DocBench} \\
\cmidrule(lr){3-4}
\cmidrule(lr){5-6}
\cmidrule(lr){7-7}
& & ACC & F1 & ACC & F1 & ACC \\
\midrule
\multirow{7}{*}{Qwen2.5-VL 7B}
& DocETL
& 22.8 & 11.8
& 23.2 & 30.7
& 28.5 \\
& M3DocRAG (top-5 cutoff)
& 22.0 & \underline{24.5}
& \underline{43.9} & \underline{46.3}
& 36.3 \\
& M3DocRAG (top-10 cutoff)
& 21.9 & 24.4
& \underline{43.9} & 46.2
& 35.8 \\
& MDocAgent
& 23.9 & 21.1
& 43.6 & 44.9
& \underline{62.2} \\
& MoLoRAG
& \underline{35.4} & 18.8
& 38.9 & 39.4
& 55.1 \\
& MHier-RAG
& 26.4 & 11.1
& 20.9 & 22.7
& 34.4 \\
& \cellcolor{LLightBlue}\modelname
& \cellcolor{LLightBlue}\textbf{37.1} & \cellcolor{LLightBlue}\textbf{24.6}
& \cellcolor{LLightBlue}\textbf{49.8} & \cellcolor{LLightBlue}\textbf{53.3}
& \cellcolor{LLightBlue}\textbf{65.5} \\
\midrule
\multirow{7}{*}{Qwen3-VL 8B}
& DocETL
& 27.6 & 15.7
& 28.1 & 35.0
& 37.9 \\
& M3DocRAG (top-5 cutoff)
& 28.8 & 32.1
& 57.1 & 58.4
& 62.3 \\
& M3DocRAG (top-10 cutoff)
& 29.9 & 33.2
& \underline{58.4} & 59.6
& 63.3 \\
& MDocAgent
& \underline{32.8} & \underline{35.6}
& \underline{58.4} & \underline{61.3}
& \underline{73.7} \\
& MoLoRAG
& 31.6 & 20.6
& 45.1 & 45.2
& 68.0 \\
& MHier-RAG
& 30.6 & 15.6
& 27.2 & 29.2
& 38.7 \\
& \cellcolor{LLightBlue}\modelname
& \cellcolor{LLightBlue}\textbf{53.2} & \cellcolor{LLightBlue}\textbf{42.3}
& \cellcolor{LLightBlue}\textbf{61.7} & \cellcolor{LLightBlue}\textbf{63.0}
& \cellcolor{LLightBlue}\textbf{75.8} \\
\midrule
\multirow{7}{*}{Qwen3-VL 32B}
& DocETL
& 27.1 & 17.2
& 28.4 & 37.2
& 41.4 \\
& M3DocRAG (top-5 cutoff)
& 42.0 & 36.4
& 61.6 & 63.4
& 65.8 \\
& M3DocRAG (top-10 cutoff)
& 33.5 & 28.8
& 53.8 & 55.1
& 43.3 \\
& MDocAgent
& \underline{49.1} & \underline{41.4}
& \underline{65.7} & \underline{66.4}
& \underline{75.8} \\
& MoLoRAG
& 44.4 & 37.7
& 61.3 & 64.3
& 72.8 \\
& MHier-RAG
& 29.9 & 15.7
& 29.8 & 32.4
& 39.8 \\
& \cellcolor{LLightBlue}\modelname
& \cellcolor{LLightBlue}\textbf{54.7} & \cellcolor{LLightBlue}\textbf{44.6}
& \cellcolor{LLightBlue}\textbf{66.2} & \cellcolor{LLightBlue}\textbf{66.9}
& \cellcolor{LLightBlue}\textbf{78.2} \\
\bottomrule
\end{tabular}
}
\vspace{-2mm}
\end{table*}

\begin{figure*}[!t]
\centering
\includegraphics[width=0.95\linewidth]{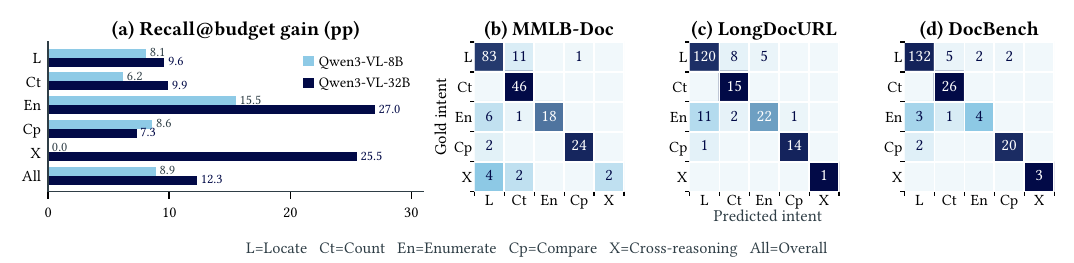}
\vspace{-4.5mm}
\caption{(a) Retrieval recall@budget gains over ColPali v1.3 by query intent on LDU under Qwen3-VL 8B and 32B. (b--d) Intent-type confusion matrices on MMLB-Doc, LDU, and DocBench under Qwen3-VL 8B.}
\label{fig:retrieval_by_intent}
\vspace{-3.5mm}
\end{figure*}

\subsubsection{Overall Accuracy}
Table~\ref{tab:overall_results} summarizes the overall accuracy comparison. \modelname{} consistently outperforms all baselines across different backbone settings, demonstrating the effectiveness of our hypergraph-based evidence construction framework.
First, the results highlight the importance of relation-preserving evidence construction. Among baselines, MoLoRAG and MDocAgent generally perform better than page-retrieval baselines such as M3DocRAG, indicating that organizing evidence through explicit relations among document content is more effective than similarity-based retrieval alone. By modeling multimodal and structural relations as document hyperedges, \modelname{} achieves substantial gains over the other baselines across the evaluated benchmarks.

Second, the consistent gains of \modelname{} further demonstrate the effectiveness of adaptive intent-aware evidence selection. \modelname{} selects evidence according to the query's required scope, while M3DocRAG relies on fixed top-$K$ page retrieval.
For M3DocRAG, moving from a top-5 cutoff to a top-10 cutoff yields only limited improvement in most settings, showing that useful evidence cannot be obtained by simply increasing a fixed retrieval budget. In the 32B setting, the top-10 variant can even perform worse, partly because the larger backbone requires tighter input compression under the same serving constraints.
This suggests that increasing retrieved input without considering the query's actual demand can add computational overhead while introducing redundant or noisy context, especially for queries that are already answerable with a smaller evidence set.
\subsubsection{Retrieval Effectiveness}
We further analyze retrieval performance under different query intents, since they impose different evidence access patterns. Figure~\ref{fig:retrieval_by_intent}(a) reports recall@budget gains by intent. For a fair comparison, we evaluate \modelname{} and ColPali v1.3~\cite{DBLP:journals/corr/abs-2407-01449} under a common budget on LDU, the only benchmark with complete page-level evidence annotations. Overall, \modelname{} improves recall@budget across most intents, with clear gains for document-scope \textsc{count} and \textsc{enumerate} queries that require aggregating evidence scattered across pages. Aggregation hyperedges help recover more complete evidence than similarity-ranked page retrieval, while fine-grained selection avoids irrelevant content even on correctly retrieved pages.
\begin{figure}[!t]
\centering
\includegraphics[width=\linewidth]{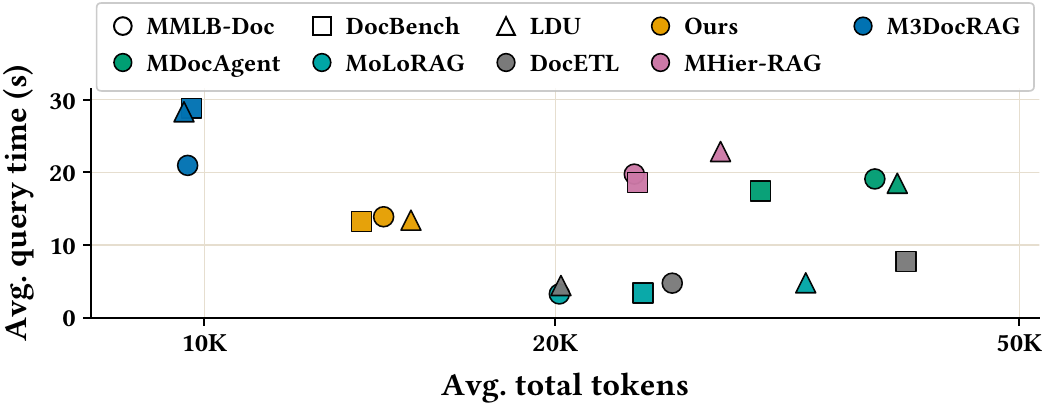}
\caption{Efficiency across baselines (Qwen3-VL 8B; lower-left is better).}
\vspace{-4mm}
\label{fig:efficiency}
\end{figure}
\subsubsection{Intent-Type Classification}
To evaluate the intent-analysis module's type predictions, we manually review 200 queries per benchmark. Figure~\ref{fig:retrieval_by_intent}(b--d) reports five-way confusion matrices for Qwen3-VL 8B, with accuracies of 86.5\%, 86.0\%, and 92.5\% on MMLB-Doc, LDU, and DocBench, respectively. To assess downstream effects, we force \textsc{count} to \textsc{locate} with other intent fields fixed, reducing QA accuracy by 7.8 and 3.1 pp on MMLB-Doc and LDU, respectively.
\begin{table}[!t]
\centering
\caption{Offline construction cost and persistent storage.}
\label{tab:offline_cost}
{\footnotesize
\setlength{\tabcolsep}{4pt}
\renewcommand{\arraystretch}{0.96}
\begin{tabular*}{0.95\columnwidth}{@{\extracolsep{\fill}}lccc@{}}
\toprule
Item & MMLB-Doc & LDU & DocBench \\
\midrule
Render + OCR (s/doc) & 120.0 & 239.6 & 184.9 \\
Hypergraph (s/doc) & 4.9 & 5.1 & 5.6 \\
Hypergraph index (MB) & 22.3 & 58.6 & 23.3 \\
OCR artifacts (MB) & 24.1 & 117.7 & 97.1 \\
Page images (GB) & 3.08 & 16.12 & 4.89 \\
\bottomrule
\end{tabular*}}
\vspace{-3.5mm}
\end{table}

\subsubsection{Overall Efficiency}
To assess end-to-end inference efficiency, we report token usage and query latency under each method's own evidence selection policy in Figure~\ref{fig:efficiency}. As shown in the figure, \modelname{} consistently stays in a compact region across the three benchmarks. Specifically, compared with MDocAgent, \modelname{} uses fewer tokens and responds faster, since most cross-modal relations are established before inference, avoiding their repeated discovery through multi-step agentic reasoning. Compared with M3DocRAG, \modelname{} substantially reduces latency, although fine-grained textual evidence can increase token usage. MoLoRAG and DocETL can be faster in some settings, but generally consume more tokens, increasing overall inference cost. Thus, \modelname{} achieves a favorable balance between token cost and query latency, making it efficient for serving scenarios where both response time and LLM token cost matter.

To complement the online analysis, Table~\ref{tab:offline_cost} reports offline construction cost and persistent storage. Rendering and OCR are preprocessing steps shared with most baselines, whereas HyperDoc-specific hypergraph construction takes 4.9--5.6 s/document and its index occupies 22.3--58.6 MB per benchmark. The overall offline pipeline scales approximately linearly with document and page counts under fixed resources.

\subsection{Ablation}
\label{sec:ablation}
\vspace{-0.8mm}
To understand the role of each component in \modelname{}, we evaluate ablated variants from both answer-quality and efficiency perspectives in Table~\ref{tab:ablation}.

First, the construction-level variants examine the effect of relation-preserving hyperedges by preserving access to primitive parsed elements while removing their relational organization. \emph{w/o Contextual HE} removes contextual hyperedges; the same parsed elements remain accessible, but the selector can no longer bind visual anchors with their associated textual context as relation-preserving evidence groups. \emph{w/o Structural HE} keeps page/section constraints for retrieval but removes explicit hierarchical relations when constructing the final evidence.
Although these variants use fewer tokens and lower latency because complete relation-level contexts are no longer rendered, both lead to clear accuracy drops across the three datasets. The larger drop of \emph{w/o Contextual HE} highlights the importance of preserving cross-modal relations, while the degradation of \emph{w/o Structural HE} shows the value of maintaining hierarchical containment during evidence construction.

Second, the selection-level variants show how online evidence selection affects the final context. Removing hyperedge selection causes the clearest efficiency degradation and also reduces answer quality on MMLB-Doc and LDU, indicating that this module is the main mechanism for filtering redundant evidence under the input budget. Reranking and pruning provide lighter refinements, where reranking improves evidence order for downstream reasoning, while pruning reduces noisy auxiliary context and token usage without breaking the structural seeds of each selected hyperedge.
Overall, the full model provides a favorable effectiveness-efficiency trade-off across the three benchmarks.

\begin{table}[t]
\centering
\footnotesize
\setlength{\tabcolsep}{3pt}
\renewcommand{\arraystretch}{0.80}
\caption{Qwen3-VL 8B ablation (HE: hyperedges; Time/Total Tok.: per-query averages).}
\label{tab:ablation}
\vspace{-1.5mm}
\begin{tabular}{@{}llcccc@{}}
\toprule
  Dataset & Variant & ACC & F1 & Time & Total Tok. \\
  \midrule

  \multirow{6}{*}{MMLB-Doc}
  & Full & \textbf{53.2} & \textbf{42.3} & 13.9s & 14.3K \\
  & w/o Contextual HE & 43.0 & 28.9 & 7.5s & 6.0K \\
  & w/o Structural HE & 48.0 & 35.0 & 10.6s & 11.1K \\
  & w/o HE Selection & 50.3 & 39.5 & 26.4s & 36.1K \\
  & w/o Reranking & 50.3 & 39.3 & 14.5s & 16.5K \\
  & w/o Pruning & 50.7 & 40.6 & 14.9s & 16.7K \\

  \midrule

  \multirow{6}{*}{LDU}
  & Full & \textbf{61.7} & \textbf{63.0} & 13.4s & 15.0K \\
  & w/o Contextual HE & 47.2 & 54.0 & 8.5s & 7.9K \\
  & w/o Structural HE & 53.8 & 58.6 & 10.3s & 12.5K \\
  & w/o HE Selection & 59.7 & 61.1 & 29.7s & 36.5K \\
  & w/o Reranking & 61.6 & 62.8 & 13.1s & 15.0K \\
  & w/o Pruning & 61.1 & 62.2 & 13.0s & 15.2K \\

  \midrule

  \multirow{6}{*}{DocBench}
  & Full &  \textbf{75.8} & -- & 13.2s & 13.6K \\
  & w/o Contextual HE & 46.4 & -- & 9.5s & 4.8K \\
  & w/o Structural HE & 59.5 & -- & 10.4s & 8.6K \\
  & w/o HE Selection & 74.5 & -- & 20.7s & 25.1K \\
  & w/o Reranking & 75.6 & -- & 12.9s & 14.1K \\
  & w/o Pruning & 75.5 & -- & 12.5s & 14.4K \\

\bottomrule
\end{tabular}
\vspace{-4mm}
\end{table}

\subsection{Parameter Sensitivity}
\label{sec:parameter_sensitivity}
\vspace{-0.8mm}

In this subsection, we study the robustness of \modelname{} to the density stopping coefficient $\alpha_{\rho}$, which determines when the selector terminates soft evidence augmentation during online inference.
\begin{table}[h]
\centering
\vspace{-1.5mm}
\footnotesize
\setlength{\tabcolsep}{3pt}
\renewcommand{\arraystretch}{0.82}
\vspace{-1.5mm}
\caption{Sensitivity analysis of $\alpha_{\rho}$ with Qwen3-VL 8B. Tok and Img denote average total tokens and selected images per query.}
\label{tab:alpha_sensitivity}
\vspace{-1.5mm}
\begin{tabular}{llcccc}
\toprule
Dataset & Metric
& $\alpha_{\rho}=0.8$
& $\alpha_{\rho}=0.9$
& $\alpha_{\rho}=1.0$
& $\alpha_{\rho}=1.1$ \\
\midrule

\multirow{4}{*}{MMLB-Doc}
& ACC & 52.40 & 52.00 & \textbf{53.20} & 52.60 \\
& F1  & 42.10 & 41.10 & \textbf{42.30} & 42.00 \\
& Tok & 17.10K & 15.79K & 14.25K & 13.63K \\
& Img & 9.81 & 8.76 & 7.52 & 7.05 \\

\midrule
\multirow{4}{*}{LDU}
& ACC & 61.20 & 61.30 & \textbf{61.70} & 61.40 \\
& F1  & 62.50 & 62.60 & \textbf{63.00} & 62.70 \\
& Tok & 19.70K & 17.28K & 15.04K & 14.25K \\
& Img & 10.95 & 9.22 & 7.55 & 7.01 \\

\midrule
\multirow{3}{*}{DocBench}
& ACC & 73.79 & 74.53 & \textbf{75.80} & 74.49 \\
& Tok & 15.95K & 14.90K & 13.64K & 14.33K \\
& Img & 8.37 & 7.49 & 6.37 & 7.14 \\

\bottomrule
\end{tabular}
\vspace{-1.5mm}
\end{table}
Table~\ref{tab:alpha_sensitivity} reports the online sweep results on MMLongBench-Doc, LongDocURL, and DocBench with $\alpha_{\rho}\in\{0.8,0.9,1.0,1.1\}$. Overall, \modelname{} is not sensitive to delicate tuning of $\alpha_{\rho}$. Across the evaluated benchmarks, varying $\alpha_{\rho}$ yields a stable accuracy-efficiency trade-off: increasing $\alpha_{\rho}$ generally reduces the average input cost, while overly aggressive stopping may slightly degrade answer quality.
We therefore set $\alpha_{\rho}=1.0$ as the default value, since it reduces unnecessary input cost while preserving competitive answer quality across the evaluated benchmarks.

\begin{figure*}[!t]
\centering
\includegraphics[width=0.92\textwidth]{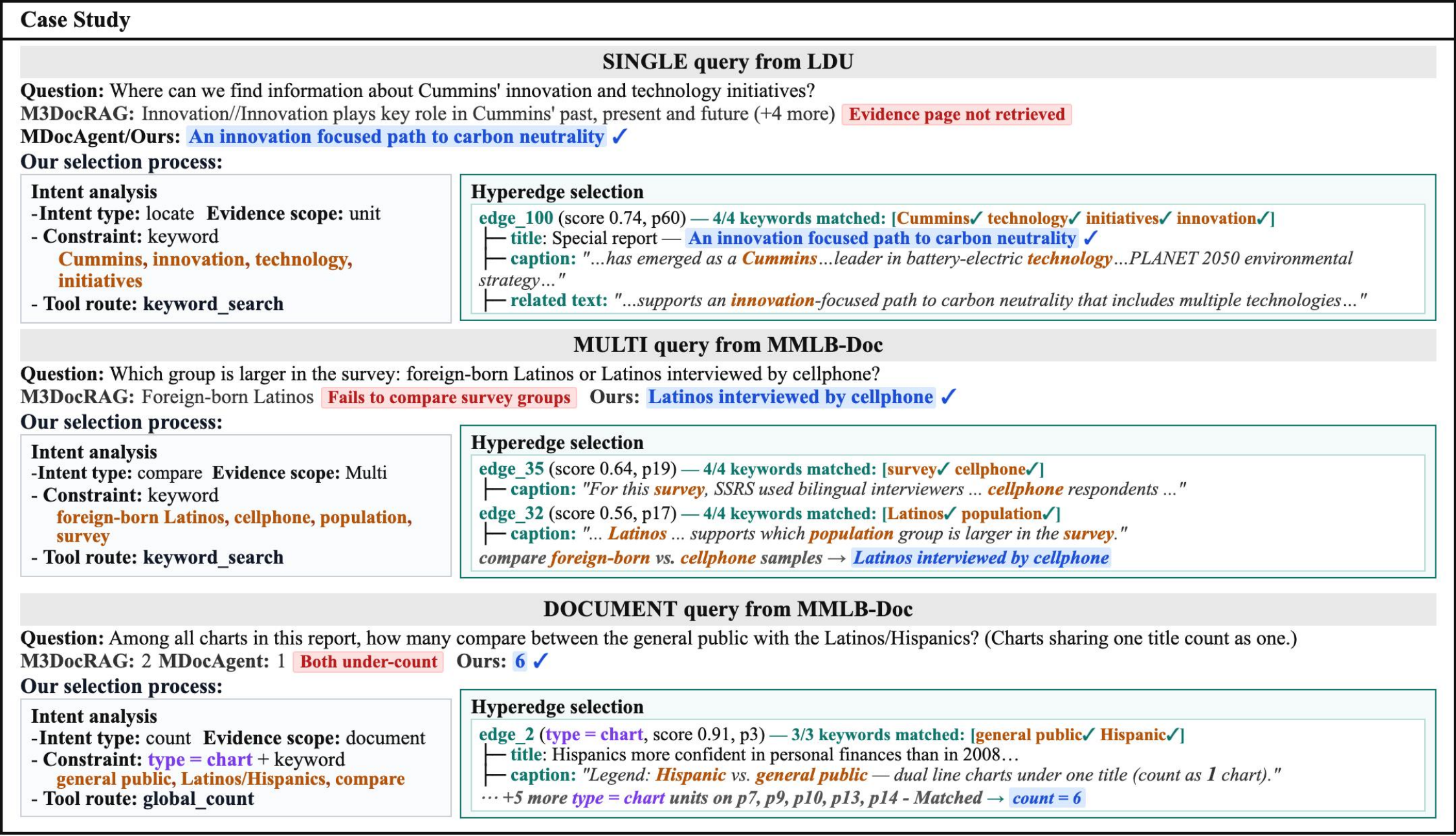}
\caption{Case study illustrating \textsc{single}-evidence localization, \textsc{multi}-evidence composition, and \textsc{document}-level aggregation under the Qwen3-VL 8B setting.}
\label{fig:case_study}
\vspace{-3mm}
\end{figure*}

\subsection{Case Study}
\label{sec:qualitative}
\vspace{-0.8mm}
To demonstrate how budget-aware evidence selection adapts to different evidence requirements, we use representative cases to examine how \modelname{} constructs query-aware evidence at inference time.
As shown in Figure~\ref{fig:case_study}, the \textsc{single} case shows that structural cues such as section titles can be essential even when they are weak under page-level similarity. \modelname{} preserves such cues by binding the target evidence with its contextual and structural relations.
The \textsc{multi} case requires composing complementary evidence units under a shared comparison constraint. Even when a relevant page is retrieved, page-centric VLM pipelines may still struggle to stitch isolated evidence fragments together and align them under the same survey context, but \modelname{} retrieves and combines the corresponding hyperedges with the goal of maximizing query constraint coverage.
The \textsc{document} case requires global aggregation over chart-type evidence, where fixed-budget retrieval under-counts because it only covers a partial set of relevant charts.
By invoking document-scope selection over matched hyperedges, \modelname{} avoids this coverage bottleneck and recovers the correct count.
Together, these cases show that \modelname{} improves evidence construction by changing the retrieval unit from isolated segments to structured evidence relations.

\vspace{-0.5mm}
\section{Related Work}
\vspace{-0.5mm}

\subsection{Document RAG}
\vspace{-0.8mm}
\newcommand{\cmark}{\checkmark}
\newcommand{\nmark}{\textcolor{gray}{\scriptsize --}}
Recent benchmarks such as MMLongBench-Doc~\cite{maMMLongBenchDocBenchmarkingLongcontext2024b} have extended document QA from earlier short-document settings to lengthy, visually rich documents that often exceed VLM input limits, making RAG a practical way to select compact evidence for downstream reasoning.
Early document QA methods mainly represent documents as plain text or OCR-parsed text segments, which are then used as retrieval units in text-based RAG pipelines.
Methods including PDFTriage~\cite{saad-falconPDFTriageQuestionAnswering2023b} further preserve the original layout and document metadata, retrieving context based on structure or content rather than text alone.
Although these approaches successfully fit long-context documents into the LLM context window, they still lose visual context that is important for figures, tables, charts, and layout-dependent questions.

With the rise of multimodal LLMs, document RAG methods have increasingly incorporated visual evidence to handle questions involving figures, tables, charts, and layout-dependent information. Page-level multimodal RAG methods retrieve relevant document pages as images for downstream VLM reasoning, improving page selection through multimodal retrieval, adaptation, reranking, or iterative refinement~\cite{choM3DocRAGMultimodalRetrieval2024b,DBLP:conf/iclr/ChenZZYDGRC025,jainSimpleDocMultiModalDocument2025,yu2025visragvisionbasedretrievalaugmentedgeneration}.
Beyond image-only retrieval, recent systems combine visual and textual evidence by gathering OCR-parsed chunks, page images, or modality-specific outputs into the final evidence set~\cite{gongMHierRAGMultiModalRAG2025,suriVisDoMMultiDocumentQA2025b,hanMDocAgentMultiModalMultiAgent2025b,dongDocResearcherUnifiedSystem2025b}.
Earlier visually rich information extraction methods exploit document layout, with VS2~\cite{Sarkhel2019VisualSegmentation} segmenting documents into logical blocks and Artemis~\cite{Sarkhel2021ImprovingIE} encoding visual spans with textual and layout context.
Despite these advances, most multimodal document RAG methods still construct retrieval around individual evidence units, such as pages, image regions, or modality-specific chunks.
Such units are often insufficient for long-document queries that require structural cues, cross-modal bindings, or evidence aggregation across multiple document regions.

To address this limitation, another line of document-centric work constructs structured evidence spaces from document hierarchy, layout, and logical dependencies.
For structured document analytics, ZenDB~\cite{linQueryingTemplatizedDocument2025} extracts semantic hierarchical trees from templatized documents and uses them as semantic indexes to support accurate and cost-effective document analytics. Along a similar database-oriented direction, QUEST~\cite{sunQUESTQueryOptimization2025} and DocETL~\cite{shankarDocETLAgenticQuery2025} leverage LLMs to generate execution plans for document processing, treating unstructured document analysis as a query optimization problem.
More recently, BookRAG~\cite{wangBookRAGHierarchicalStructureaware2025} transforms raw documents into structured hierarchical trees and further populates a knowledge graph by extracting and refining entities from tree nodes.
For visually rich document QA, MoLoRAG constructs a page-level graph using a similarity threshold and performs VLM-guided graph traversal for retrieval~\cite{wuMoLoRAGBootstrappingDocument2025}.
These studies show that explicit document structures are useful for organizing evidence, but they are typically used for hierarchy-aware retrieval, planning, or page-level graph traversal, rather than for selecting multimodal evidence under query-specific input budgets. As summarized in Table~\ref{tab:evidence_construction_aspects}, most existing methods rely on fixed top-$K$ retrieval cutoffs or LLM-based selection policies, whereas \modelname{} formulates evidence selection over a structured hypergraph and controls the input budget through a cost-aware greedy algorithm.
\vspace{-0.8mm}
\subsection{Structured RAG}
Recent studies have explored graphs and knowledge graphs as retrieval indexes for RAG, with GraphRAG~\cite{edgeLocalGlobalGraph2025b,hanRetrievalAugmentedGenerationGraphs2025b,yangSuperRAGRAGLayoutAware2025,zhaovog,zhuOpenMELUnsupervisedMultimodal2025d} supporting retrieval over entities, relations, neighborhoods, and communities while preserving connectivity beyond independent chunks.
Similarly, Docs2KG constructs knowledge graphs from heterogeneous documents and web pages~\cite{sunDocs2KGUnifiedKnowledge2024}.
These methods show that explicit relational indexes can improve evidence organization and query-time retrieval coverage~\cite{wangbeyond}.
However, ordinary graph indexes mainly encode pairwise edges, which can be insufficient when the natural retrieval unit is a group of mutually dependent items. Recent works on knowledge hypergraphs and hypergraph-enhanced RAG~\cite{Lee_2020,yang2023hgmatchmatchbyhyperedgeapproachsubgraph, 8509253, fengHyperRAGCombatingLLM2026,luoHyperGraphRAGRetrievalAugmentedGeneration2025} show that higher-order structures can be represented as hyperedges connecting multiple nodes, benefiting retrieval and reasoning over complex relational knowledge.
These studies demonstrate the benefit of higher-order indexing for open-domain RAG and suggest its potential for organizing multimodal evidence in long-context document QA.
\begin{table*}[!t]
\centering
\footnotesize
\setlength{\tabcolsep}{3pt}
\renewcommand{\arraystretch}{0.80}
\caption{Comparison of evidence construction and selection strategies in existing document RAG methods.}
\label{tab:evidence_construction_aspects}
\vspace{-2mm}
\begin{tabular}{>{\raggedright\arraybackslash}p{0.18\textwidth}
>{\centering\arraybackslash}p{0.18\textwidth}
>{\centering\arraybackslash}p{0.20\textwidth}
>{\centering\arraybackslash}p{0.20\textwidth}
>{\centering\arraybackslash}p{0.18\textwidth}}
\toprule
{\centering\arraybackslash\textbf{Method}}
& {\centering\arraybackslash\textbf{Evidence Unit}}
& {\centering\arraybackslash\textbf{Hierarchical Structure}}
& {\centering\arraybackslash\textbf{Cross-modal Binding}}
& {\centering\arraybackslash\textbf{Selection Policy}} \\
\midrule

\multicolumn{5}{l}{\textit{Independent evidence construction}} \\
M3DocRAG~\cite{choM3DocRAGMultimodalRetrieval2024b}  & Page & \nmark & \nmark & Top-$K$ cutoff \\
MDocAgent~\cite{hanMDocAgentMultiModalMultiAgent2025b} & Page & \nmark & \cmark & Top-$K$ cutoff \\
\midrule
\multicolumn{5}{l}{\textit{Structured evidence construction}} \\
QUEST~\cite{sunQUESTQueryOptimization2025}    & Segment & \cmark & \nmark & LLM-generated plan \\
DocETL~\cite{shankarDocETLAgenticQuery2025}   & Chunk & \cmark & \nmark & LLM-generated plan \\
ZenDB~\cite{linQueryingTemplatizedDocument2025}    & Tree Node & \cmark & \nmark & Predefined template \\
BookRAG~\cite{wangBookRAGHierarchicalStructureaware2025}  & Tree / KG Node & \cmark & \cmark & LLM-generated plan \\
MHier-RAG~\cite{gongMHierRAGMultiModalRAG2025} & Page / Segment & \cmark & \cmark & Top-$K$ cutoff \\
MoLoRAG~\cite{wuMoLoRAGBootstrappingDocument2025}  & Graph Node & \nmark & \cmark & Top-$K$ cutoff \\

\midrule
\cellcolor{LLightBlue}\textbf{Ours}
& \cellcolor{LLightBlue}\textbf{Hyperedge}
& \cellcolor{LLightBlue}\cmark
& \cellcolor{LLightBlue}\cmark
& \cellcolor{LLightBlue}Cost-aware selection \\
\bottomrule
\end{tabular}
\vspace{-3.5mm}
\end{table*}

\vspace{-1mm}
\section{Conclusion}
In this paper, we formulate long-context multimodal document QA as a budget-aware evidence management problem and propose \modelname{}, a document RAG framework for adaptive evidence selection over lengthy multimodal content.
Specifically, with a document hypergraph preserving structural containment and cross-modal relations among multi-granular elements, \modelname{} formulates evidence selection as an NP-hard coverage maximization problem and employs a cost-aware greedy process that adapts to each query's demand.
Across representative benchmarks, \modelname{} consistently improves answer quality over diverse baselines with a favorable cost-latency trade-off, providing a unified pipeline for scalable multimodal document analytics.

\section{Acknowledgments}
Lei Chen's work is partially supported by National Key Research and Development Program of China Grant No. 2023YFF0725100, National Science Foundation of China (NSFC) under Grant No. U22B2060, Guangdong-Hong Kong Technology Innovation Joint Funding Scheme Project No. 2024A0505040012, the Hong Kong RGC GRF Project 16213620, RIF Project R6020-19, AOE Project AoE/E-603/18, Theme-based project TRS T41-603/20R, CRF Project C2004-21G, Key Areas Special Project of Guangdong Provincial Universities 2024ZDZX1006, Guangdong Province Science and Technology Plan Project 2023A0505030011, Guangzhou municipality big data intelligence key lab, 2023A03J0012, Hong Kong ITC ITF grants MHX/078/21 and PRP/004/22FX, Hong Kong ITC TC-SKLCRCC26EG01, Zhujiang scholar program 2021JC02X170, Microsoft Research Asia Collaborative Research Grant, HKUST-Webank joint research lab, 2025 HKUST Shenzhen-Hong Kong Collaborative Innovation Institute Green Sustainability Special Fund from Shui On Xintiandi and the InnoSpace GBA, and HKUST(GZ)-CMCC(Guangzhou Branch) Metaverse Joint Innovation Lab under Grant No. P00659.
Jiachuan Wang's work is supported in part by JST CREST (JPMJCR22M2). Jiake Ge's work is supported by the National Natural Science Foundation of China (No.62602465). Haoyang Li's work is supported by Guangdong Basic and Applied Basic Research Foundation 2026A1515010227, PolyU TDG25-28/TBP/11-IICA, PolyU PGMS (P0052504, P0053707, P0061736, P0064506, P0064525, P0064570, P0057770). Yongqi Zhang's work is supported by Guangdong Basic and Applied Basic Research Foundation 2025A1515010304, Guangdong Province Project 2024QN11X088, Guangzhou Science and Technology Planning Project 2025A03J4491. 
 \newpage
\section*{AI-Generated Content Disclosure}
OpenAI Codex was used as an assistive tool for manuscript polishing and formatting, as well as for code execution, debugging, and visualization in the experiments reported in Section~\ref{sec:experiments}. All AI-assisted outputs were reviewed and, where appropriate, revised by the authors to ensure accuracy and consistency with the underlying data, methods, and intended presentation. The authors remain fully responsible for the integrity and originality of the scientific content.

\bibliographystyle{IEEEtran}
\bibliography{references}

\end{document}